\documentclass[bj,authoryear]{imsart}

\RequirePackage{graphicx}
\usepackage{caption}
\usepackage{subcaption}
\usepackage{float}
\usepackage{multirow}

\startlocaldefs
\theoremstyle{plain}

\newtheorem{theorem}{Theorem}[section]

\newtheorem{corollary}[theorem]{Corollary}

\theoremstyle{remark}
\newtheorem{definition}[theorem]{Definition}

\endlocaldefs

\begin{document}

\begin{frontmatter}
\title{Multi-Quantile Estimators for the Parameters of Generalized Extreme Value Distribution}
\runtitle{Quantile-based estimators for $\xi \in \mathbb R$}

\begin{aug}
\author[A]{\fnms{Sen}~\snm{Lin}\thanks{{Corresponding author.}}\ead[label=e1]{slin31@uh.edu}}
\author[B]{\fnms{Ao}~\snm{Kong}\ead[label=e2]{aokong@nufe.edu.cn}
}
\author[A]{\fnms{Robert}~\snm{Azencott}\ead[label=e3]{robertazencott@gmail.com}}

\address[A]{Department of Mathematics,
University or Houston\printead[presep={,\ }]{e1,e3}}

\address[B]{School of Finance,
Nanjing University of Finance and Economics\printead[presep={,\ }]{e2}}
\end{aug}

\begin{abstract}
We introduce and study Multi-Quantile estimators for the parameters $\xi, \sigma, \mu$ of Generalized Extreme Value (GEV) distributions to provide a robust approach to extreme value modeling. Unlike classical estimators, such as the Maximum Likelihood Estimation (MLE) estimator and the Probability Weighted Moments (PWM) estimator, which impose strict constraints on the shape parameter $\xi$, our estimators are always asymptotically normal and consistent across all values of the GEV parameters. The asymptotic variances of our estimators decrease with the number of quantiles increasing and can approach the Cramér-Rao lower bound very closely whenever it exists. Our Multi-Quantile estimators thus offer a more flexible and efficient alternative for practical applications. We also discuss how they can be implemented in the context of Block Maxima method.
\end{abstract}



\end{frontmatter}

\section{Introduction}
For long sequence $X_1, X_2, \ldots, X_n$ of i.i.d. random variables, the distributions of extreme values $Y_n = \max(X_1, \ldots, X_n)$ are often modeled, after adequate affine rescaling, by the well-known family of Generalized Extreme Value (GEV) distributions denoted as $G_{\theta}$, parameterized by $\theta = (\xi, \mu, \sigma) \in \mathbb{R}^2 \times \mathbb{R}^+$. Recall that the Cumulative Distribution Function (CDF) $G_{\theta}$ is given by:
\begin{equation*}
G_{\theta}(y) =
\begin{cases}
\exp\left(-\left(1 + \xi\left(\frac{y - \mu}{\sigma}\right)\right)^{-1/\xi}\right), & \xi \neq 0 \\
\exp\left(-\exp\left(-\frac{y - \mu}{\sigma}\right)\right), & \xi = 0
\end{cases}
\end{equation*}
so that each $G_{\xi, \mu, \sigma}$ is derived from $G_{\xi, 0, 1}$ by centering and rescaling.

A cumulative distribution function $F$ belongs to the domain of attraction $ D(\xi)$ of $G_{\xi, \mu, \sigma}$ if one can find sequences $a_n > 0$ and $b_n$ such that for any i.i.d. sequence $X_1, \ldots, X_n,\ldots$ with a common distribution $F$, then $Z_n = \left(\max\{X_1, \ldots, X_n\} - b_n\right)/a_n$ converges in distribution to $G_{\xi, \mu, \sigma}$. The set $D(\xi)$ clearly does not depend on the location and scale parameters $(\mu, \sigma)$. Moreover $D(\xi)$ has been fully characterized by classical results (see below and \cite{fisher1928},\cite{gnedenko1943}). For example, the Gaussian and exponential distribution belong to $D(0)$, whereas the uniform distribution on an interval belongs to $D(-1)$.

Given an i.i.d. sequence $X_1, \ldots, X_n, \ldots$ with an unknown common distribution $F$, assumed to belong to $D(\xi)$ for some unknown $\xi$, accurate estimation of $\xi$ is often a key question in practical modeling of extreme events in stock markets , climate evolution , hydrology analysis, etc (\cite{Lin2024}).

Several publications have proposed estimators $\hat{\theta}_n$ of $\theta = (\xi, \mu, \sigma)$. Explicit asymptotic distributions for $\sqrt{n}(\hat{\theta}_n - \theta)$ have been derived in a few papers, but the associated asymptotic results always require strong restrictions on the shape parameter $\xi$. Let us recall the main previous rigorous asymptotic results.
\begin{enumerate}
    \item Maximum Likelihood Estimation (MLE) method: After initial studies by \cite{prescott1980} and \cite{hosking1985a}, asymptotic normality of MLE was confirmed by \cite{bucher2017}, but \emph{only under the constraint $\xi > -1/2$}. Indeed when $\xi\leq -1/2$, local maxima of the log-likelihood do not exist. 
    \item Probability Weighted Moments (PWM) method: After PWM estimators were introduced  by \cite{hosking1985b}, their asymptotic normality was proved by \cite{ferreira2015}, but \emph{only under the constraint $\xi < 1/2$}. Indeed when $\xi\geq 1/2$, the second order moment of $G_{\xi, \mu, \sigma}$ does not exist.
\end{enumerate}

{Due to the constraints on $\xi$, it is not possible to determine the appropriate estimator in advance, as the value of $\xi$ is unknown.} In this paper, we introduce and study new \emph{Multi-Quantile} (MQ) estimators  $\hat{\theta}_n$ of $\theta = (\xi, \mu, \sigma)$, based on any fixed number $k$ of empirical quantiles. These estimators are asymptotically normal for all $(\xi, \mu, \sigma)$, \emph{without any restriction on the shape parameter $\xi$}. The asymptotic $3\times 3$ covariance matrix of $\hat\theta_n$ is explicitly computable. Moreover, as $k$ increases, the asymptotic variance of $\hat\xi_n$ decreases and tends to the optimal Cramer-Rao bound whenever this bound exists. For parameter estimation of Generalized Pareto distribution, Multi-Quantile estimators were studied by \cite{castillo1997}. 

Our theoretical results are presented in Sections \ref{section: 3q}, \ref{section: mq}, and \ref{section: efficient}. Section \ref{section: compare} provides a detailed analysis comparing the theoretical and empirical accuracy of our MQ estimator with two previously studied estimators (MLE, PWM). Section \ref{section: block maxima} focuses on combining our MQ estimators with the Block Maxima (BM) setup introduced by \cite{gumbel1958} and studied recently in \cite{ferreira2015}, \cite{dombry2015}, \cite{dombry2019}, \cite{padoan2024}. We again prove asymptotic normality for our MQ estimators in the BM setup.
\section{Asymptotic normality of empirical quantiles}\label{section: notation}
Let $X_1, \ldots, X_N$ be an i.i.d. random sample. Let $F(x)$ be the CDF of the $X_j$. Assume that the support of $F$ is a (possibly infinite) interval $(a,b)$, and that for $a<x<b$, $F(x)$ has a positive and continuous density $f(x) = F'(x)$. Fix any \emph{not necessarily ordered} set of $k$ \emph{distinct} percentiles $\boldsymbol{q}=\left[q_1,\dots,q_k\right]$ with $0<q_i<1$. Denote $T_i$ the true $q_i$-quantile of $F$ and $\hat{T}_i$ the empirical $q_i$-quantile of $X_1, \ldots, X_N$. Let $\boldsymbol{T(q)}=\left[T_1,\dots,T_k\right]$ and $\boldsymbol{\hat{T}(q)}=\left[\hat{T}_1,\dots,\hat{T}_k\right]$.

 Then (see Cor. 21.5 in \cite{vandervaart1998}) the random vector $\sqrt{N}\left(\boldsymbol{\hat T(q)-T(q)}\right)$ is  asymptotically normal with mean $0$ and asymptotic covariance matrix $\kappa$ given by:
\begin{equation}\label{eq:Mij}
\kappa_{i,j}=\kappa_{j,i}=\frac{\max(q_i,q_j)-q_iq_j}{f(T_i)f(T_j)} \text{ for all }i,j=1\dots k
\end{equation}
when $N\to\infty$.
Note that since the $q_i$ are distinct , the matrix $\kappa$ is always invertible. 
We now apply this generic result when $F$ is a GEV distribution $G_\theta$. For any percentile $0<q_i<1$, the $q_i$-quantile $T_i$ of $G_\theta$ is given by:
\begin{align}\label{eq:quantile}
T_i=\begin{cases}
    \mu + \frac{\sigma}{\xi} \left(\exp(-\xi LL_i)-1\right) &\text{when}\ \xi\neq 0 \\
    \mu - \sigma LL_i &\text{when}\ \xi = 0
\end{cases}
\end{align}
where $LL_i=\log(-\log(q_i))$.

The support $spt(\theta)$ of $G_\theta$ is given by:
\begin{align*}
spt(\theta)&=
\begin{cases}
    (\mu-\sigma/\xi , +\infty ) &\text{for}\ \xi > 0\\
(-\infty , \mu-\sigma/\xi ) &\text{for}\ \xi < 0 \\
(-\infty , +\infty) &\text{for}\ \xi = 0
\end{cases}
\end{align*}

When $\xi\neq 0$, the density $g_\theta(x)$ of $G_\theta$ is given by, for all $x\in spt(\theta)$:
\begin{align}
g_{[\xi,0,1]}(x) &= (1+\xi x)^{-1+1/\xi} \exp\left(-(1+\xi x)^{-1/\xi}\right) \\
g_{[\xi,\mu,\sigma]}(x) &= \frac{1}{\sigma} g_{[\xi,0,1]}\left(\frac{x - \mu}{\sigma}\right) \label{eq:gtheta}
\end{align}

When $\xi=0$, one has for all $x \in \mathbb{R}$:
\begin{equation}
g_{[0,\mu,\sigma]}(x) = \exp\left(-\frac{x - \mu}{\sigma}\right)\exp\left(-\exp\left(-\frac{x - \mu}{\sigma}\right)\right)
\end{equation}

Fix any \emph{not necessarily ordered} set of \emph{distinct} percentiles $\boldsymbol{q}=[q_1,\dots,q_k]$. Let $\boldsymbol{T(q)}$ be the associated vector of  true quantiles of  $G_\theta$. Given $N$ i.i.d. observations sampled from $G_\theta$, denote $\boldsymbol{\hat{T}(q)}$ the vector of empirical quantiles defined by $\boldsymbol{q}$. Since $g_\theta(x)= G'_\theta(x)>0$ for all $x \in spt(\theta)$, asymptotic normality of empirical quantiles will hold as $N\to\infty$. More precisely, $\sqrt{N}(\boldsymbol{\hat{T}(q)} - \boldsymbol{T(q)}) $ is  asymptotically normal, with asymptotic mean $0$ and covariance matrix $\Sigma_T$ given by
\begin{equation}\label{eq:Sigma.k}
\Sigma_T(i,j)=\sigma^2 \frac{\max(q_i,q_j)-q_iq_j}{q_iq_j(\log(q_i)\log(q_j))^{1+\xi}} \;\;
\text{for all}\; i,j = 1\dots k
\end{equation}
This equation is derived from \eqref{eq:Mij}, after computing $g_\theta(T_i), g_\theta(T_j)$ using equations \eqref{eq:quantile} and \eqref{eq:gtheta}. As already noted above, since the $q_i$ are distinct, the $k \times k$ matrix $\Sigma_T$ is invertible.

\section{Three-Quantiles estimators of GEV parameters }\label{section: 3q}
We introduce a family of asymptotically normal estimators $\hat{\theta}$ for the parameter vector $\theta = (\xi, \mu, \sigma)$ of a GEV distribution $G = G_{\theta}$. Each such $\hat{\theta}$ will be easily computable in terms of three empirical quantiles. We deliberately restrict the proofs and detailed presentation to the generic situation $\xi \neq 0$. The simpler case $\xi = 0$ is analyzed separately whenever necessary.

\subsection{Relations between GEV quantiles and GEV parameters estimators}
\begin{theorem}\label{th:3q}
    Fix any $\theta = [\xi, \mu, \sigma] \in \mathbb{R}^3$ with $\xi \neq 0$ and $\sigma > 0$. For any triplet of percentiles $\boldsymbol{q} = [q_1, q_2, q_3]$ with $0 < q_1 < q_2 < q_3 < 1$, denote $T_j$  the $q_j$-quantile of the GEV $G_{\theta}$. Then $\theta$ is uniquely determined by the vector $\boldsymbol{T} = [T_1, T_2, T_3]$. In fact, one has $\theta = H(\boldsymbol{T})$, where $H$ is a $C^{\infty}$ function of $\boldsymbol{T}$ which can be easily computed as follows:
    
    Denote $a_1 = \log\left(\frac{\log(q_1)}{\log(q_3)}\right)$, $a_2 = \log\left(\frac{\log(q_2)}{\log(q_3)}\right)$, and $b = \frac{T_3 - T_2}{T_3 - T_1}$.
    
    Then $\xi = \Phi(\boldsymbol{T})$ is the unique non-zero solution of $h(x) = 0$, where $h(x)$ is given by:
    
    \begin{equation}
    \label{eq:hxi}
    h(x) = \exp(-x a_2) - b \exp(-x a_1) + b \text{ for all } x \in \mathbb{R}.
    \end{equation}
    
    The location and scale parameters $\mu$ and $\sigma>0$ are then  explicit smooth functions $\mu = L(\boldsymbol{T})$ and $\sigma = S(\boldsymbol{T})$ of $\boldsymbol{T}$.
\end{theorem}

\begin{proof}
For $j = 1, 2, 3$, denote
\begin{equation} \label{eq:LLjQj}
LL_j = \log(-\log(q_j)), \quad Q_j = \frac{1}{\xi}\left(\exp(-\xi LL_j) - 1\right)
\end{equation}
so that $LL_1 > LL_2 > LL_3$ and $Q_1 < Q_2 < Q_3$. Then equation \eqref{eq:quantile} directly shows that $(\mu, \sigma)$ verify the linear system
\begin{equation}
T_j - \mu - Q_j \sigma = 0 \quad \text{for} \; j = 1, 2, 3
\label{eq:nonlinearsys}
\end{equation}
This system of three linear equations is verified by the non-zero vector $[1, \mu, \sigma]$. Hence, the $3\times3$ matrix $K$ of system coefficients must verify $\det(K) = 0$, which yields
\begin{equation}
(T_2 - T_3)Q_1 + (T_3 - T_1)Q_2 + (T_1 - T_2)Q_3 = 0
\end{equation}
Replace $Q_j$ by the formula given by \eqref{eq:LLjQj} to get the following identity, valid for all $\xi$,
\begin{equation}\label{eq:triplet}
(T_2 - T_3) \exp(-\xi LL_1) + (T_3 - T_1) \exp(-\xi LL_2) + (T_1 - T_2) \exp(-\xi LL_3) = 0
\end{equation}

Let $a_1 = LL_1 - LL_3$, $a_2 = LL_2 - LL_3$, $b = (T_3 - T_2)/(T_3 - T_1)$, so that $0 < B < 1$ and $a_1 > a_2 > 0$. Then \eqref{eq:triplet} is equivalent to $h(\xi) = 0$, where
\[
h(x) = \exp(-x a_2) - b \exp(-x a_1) - 1 + b = 0
\]

For $\boldsymbol{q}$ fixed, $h(x)$ is a $C^{\infty}$ function of $x$ and $b$, and has derivative
\begin{equation}\label{equation: dh}
    h'(x) = a_1 b  \exp(-x a_2)\left[\exp(-x (a_1 - a_2)) - \frac{a_2}{a_1 b}\right]
\end{equation}

Since $h'(x)$ has the same sign as $(x - s)$ with $s = \frac{\log(a_1 b/a_2)}{a_1 - a_2}$, $h(x)$ is increasing on $(-\infty, s)$ and decreasing on $(s, +\infty)$, with $h(-\infty) = -\infty$ and $h(+\infty) = b - 1 < 0$. Thus, the equation $h(x) = 0$ has exactly two solutions: namely $x = 0$ and a unique non-zero solution $\xi = \psi(b) \neq s$. If $a_1 b > a_2$, then $s > 0$, and $\xi > s > 0$. If $a_1 b < a_2$, then $s < 0$, and $\xi < s < 0$.

The only root of $h'(x) = 0$ is $x = s$, and $\xi \neq s$, hence $h'(\xi)$ must be non-zero. So the Inverse Function theorem (see \cite{hamilton1982inverse}) applies to the $C^{\infty}$ function $h(x)$ to prove that $\xi = \psi(b)$ is a $C^{\infty}$ function of $b$, and a fortiori $\xi = \Phi(\boldsymbol{T})$ where $\Phi$ is a $C^{\infty}$ function of $\boldsymbol{T}$.

For $\xi = \Phi(\boldsymbol{T})$, the three equations of the system \eqref{eq:nonlinearsys} are linearly dependent and uniquely determine $\sigma,\mu$, by the formulas:
\begin{equation*}
    \sigma = \frac{T_2 - T_1}{Q_2 - Q_1}; \quad \mu = \frac{T_1Q_2 - Q_1T_2}{Q_2 - Q_1}
\end{equation*}

Since $Q_1,Q_2$ are explicit smooth functions of $\xi=\Phi(\boldsymbol{T})$, these expressions are clearly $C^{\infty}$ functions of $\boldsymbol{T}$, which we denote $\mu=L(\boldsymbol{T})$, and $\sigma=S(\boldsymbol{T})$.  This concludes the proof.
\end{proof}

\subsection{Numerical computation of the function $H=[\Phi,L,S]$}\label{Hcomputation}

Fix $\boldsymbol{q}=[q_1,q_2,q_3]$ with $0 < q_1 < q_2 < q_3 < 1$. For any vector $\boldsymbol{T}=[T_1,T_2,T_3]$ in $\mathbb{R}^3$ such that $T_1 < T_2 < T_3$, the proof of  the preceding theorem indicates  a fast numerical computation of the smooth functions $H(\boldsymbol{T})=\left[\Phi(\boldsymbol{T}),L(\boldsymbol{T}), S(\boldsymbol{T})\right]$. Indeed, $\boldsymbol{q}$ and $\boldsymbol{T}$ define the coefficients $a_1,a_2,b$ of the equation $h(x) = 0$.  The key numerical first step is to solve $h(x)=0$ for $x$ in an explicit half-line $J$ disjoint from $0$, with $h(x)$ strictly monotonous over $J$. The Newton-Raphson algorithm yields the solution $\xi= \Phi(\boldsymbol{T})$. This yields the values  of $Q_1,Q_2$, and then $L(\boldsymbol{T}), S(\boldsymbol{T})$ are given by $(T_2 - T_1)/(Q_2 - Q_1)$ and $(T_1Q_2 - Q_1T_2)/(Q_2 - Q_1)$. 

\subsection{Three-Quantiles estimators of GEV parameters}
Since the triplet of quantiles $\boldsymbol{T}$ can be  estimated from observations, theorem \ref{th:3q} leads us to  define the following natural class of estimators for $\theta$. 

\begin{definition}[Three-Quantiles estimator] \label{def:3quantiles}
Fix any 3 percentiles $\boldsymbol{q}=[q_1, q_2,q_3]$ with $0<q_1<q_2<q_3<1$. Let $Y_1,\ldots,Y_N$ be $N$ i.i.d. observations sampled from a GEV distribution $G_{\theta}$, with unknown $\theta=[\xi,\mu,\sigma]$. Let $\hat{T}_j$ be the empirical $q_j$-quantile of $Y_1,\dots,Y_N$, and let $\hat{\boldsymbol{T}}=[\hat{T}_1,\hat{T}_2,\hat{T}_3]$. Define the \textit{three-quantiles estimator} $\hat{\theta}_N(\boldsymbol{q})$ of the unknown $\theta$ by 
\begin{equation}
\hat{\theta}_N(\boldsymbol{q}) = H(\boldsymbol{\hat{T}})
\end{equation}
where the smooth function $H(\boldsymbol{T})$ is computed as indicated in Section \ref{Hcomputation}.
\end{definition}

\begin{theorem}[Asymptotic Normality of Three-Quantiles estimators]  \label{th:3q normality}
Fix any vector $\theta=[\xi,\mu,\sigma]$ of GEV parameters, with $\xi\neq0$ and $\sigma>0$. Fix any 3 percentiles $\boldsymbol{q}=[q_1, q_2,q_3]$ with $0<q_1<q_2<q_3<1$. Let $\hat{\theta}_N(\boldsymbol{q})$ be the \textit{Three-Quantile estimator} of the unknown $\theta$ defined above  by $\boldsymbol{q}$ and  $N$  i.i.d.  observations $Y_1,\ldots,Y_N$ sampled from $G_{\theta}$. Then $\hat{\theta}_N(\boldsymbol{q})$ is an asymptotically normal and consistent estimator of $\theta$ as $N\to \infty$. The asymptotic covariance matrix $\Gamma$ of $\sqrt{N}(\hat{\theta}_N(\boldsymbol{q}) - \theta)$ is easily computable as outlined in the proof below, where we give detailed formulas for the asymptotic variances $avar(\xi), avar(\sigma), avar(\mu)$ of  $\sqrt{N}(\hat{\xi}_N(\boldsymbol{q}) - \xi)$,  $\sqrt{N}(\hat{\sigma}_N(\boldsymbol{q}) - \sigma)$,  $\sqrt{N}(\hat{\mu}_N(\boldsymbol{q}) - \mu)$.  These  explicit formulas compute the diagonal terms  of  $\Sigma_{\theta}$ which are  essential  to derive confidence intervals for our three-quantiles estimators $\hat\xi$,  $\hat\sigma$, $\hat\mu$  of $\xi , \sigma, \mu$ .

\end{theorem}

\begin{proof}
As $N\to \infty$, the vector $\boldsymbol{\hat{T}}_N$ of 3 empirical quantiles extracted from the observed i.i.d sample $Y_1,\dots,Y_N$  is an asymptotically normal and consistent estimator for the true quantiles vector $\boldsymbol{T}$  of $G_{\theta}$. Moreover the asymptotic covariance matrix  $\Sigma_T$ of $\sqrt{N}(\boldsymbol{\hat{T}}_N - \boldsymbol{T})$ is explicitly given by formula \eqref{eq:Sigma.k}, which we recall here  
\begin{equation}\label{eq:Sigma3}
\Sigma_T(i,j)=\sigma^2 \frac{\max(q_i,q_j)-q_iq_j}{q_iq_j(\log(q_i)\log(q_j))^{1+\xi}} \;\;
\text{for all}\; i,j = 1\dots 3
\end{equation}

Since $\hat{\theta}_N(\boldsymbol{q})$ is a $C^\infty$ smooth function  $H(\boldsymbol{\hat{T}}_N)$ of $\boldsymbol{\hat{T}}_N$, the classical Cramer-Wold theorem in \cite{cramer1946} implies that $\hat{\theta}_N(\boldsymbol{q})$ must also be an asymptotically normal and consistent estimator of $H(\boldsymbol{T})=\theta$, with asymptotic covariance matrix $\Gamma = D\Sigma_T D^*$, where the $3\times 3$ matrix  $D$ is the differential  $D=\partial_{\boldsymbol{T}} H(\boldsymbol{T})$ of the smooth function $\boldsymbol{T}\to H(\boldsymbol{T})$ with respect to $\boldsymbol{T}$.

For large N, denote all normalized errors of estimation by 
$$\Delta \boldsymbol{T} = \sqrt{N}(\hat{\boldsymbol{T}}-\boldsymbol{T}) \;\;;\;
\Delta \xi = \sqrt{N}(\hat{\xi}- \xi) \; ;  \; \Delta \sigma = \sqrt{N}(\hat{\sigma}- \sigma) \;\; ; \; \Delta \mu = \sqrt{N}(\hat{\mu}- \mu) 
$$
As just proved the random vector $ [ \Delta \xi  , \;  \Delta \sigma \ ; , \Delta \mu] $ is asymptotically Gaussian when $N$ tends to $\infty$. We now outline the computation of its asymptotic covariance matrix $\Gamma = D\Sigma_T D^*$.  Fix three percentiles  $\boldsymbol{q}= \{ 0 < q_1 < q_2 < q_3 < 1] $. Let $T_j$ be the true $q_j$-quantile of $G_{\theta}$ and as in theorem \eqref{th:3q}, let  $a_1=LL1-LL3$, $a_2= LL2-LL3$, $b = (T_3-T_2)/(T_3-T_1)$. Recall that the true $\xi$ is the unique non-zero solution $Phi(\boldsymbol{T(q)})$ of equation $h(x) =0$ (see \eqref{eq:hxi}), which we now rewrite
\begin{equation*}
    h(x) = \kappa(x,\boldsymbol{T(q)}) = \exp(-x a_2) - b \exp(-x a_1) - 1 + b
\end{equation*}
The identity $\kappa(\Phi(\boldsymbol{T(q)}),\boldsymbol{T(q)})=0$  implies
  \begin{equation}\label{eq:gradPhi}
\partial_{T_j}\Phi=-\partial_{T_j}\kappa/\partial_x\kappa
\end{equation}
This  yields the  explicit expression 
\begin{equation} \label{eq:AVAR}
avar(\xi) = W(\boldsymbol{q}) \Sigma_T W(\boldsymbol{q})^* 
\end{equation}
where $\Sigma_T$ is given by \eqref{eq:Sigma3}, and the gradient vector $W(\boldsymbol{q}) =  \partial_{\boldsymbol{T}}\Phi$  is given by 
\begin{equation}\label{eq:W(q)} 
    W(\boldsymbol{q})= \alpha V 
\end{equation}
with $\alpha$ and $V$  computed by
\begin{align*}
\alpha = \frac{\exp(-\xi  a_1)-1}{-a_2\exp(-\xi a_2) + b a_1 \exp(-\xi a_1)} \;\; and \;\; V= \frac{1}{(T_3-T_1)^2}\left[T_3-T_2,\; T_1-T_3,\;T_2-T_1\right] 
\end{align*}
We have defined above the column vector $\boldsymbol{Q}= [Q_1, Q_2, Q_3]^*$ where $Q_j = \frac{\exp(-\xi  LL_j) - 1}{\xi }$.This directly yields the vector $\partial_{\xi} \boldsymbol{Q} $ :
$$
\partial_{\xi} Q_i = \frac{1 -(1+\xi LL_i)\exp(-\xi LL_i) }{\xi ^2}\\
$$
Replacing true quantiles by empirical quantiles  in the expression of $\boldsymbol{Q}$ yields a consistent estimator $\hat{\boldsymbol{Q}}$ of $\boldsymbol{Q} $.\\
Then  $\Delta_Q = \sqrt{N}(\hat{\boldsymbol{Q}}-\boldsymbol{Q})$ converges in distribution to a gaussian $N(0,\Sigma_Q)$. Again by Cramer-Wold theorem we must have
\begin{equation}
\Sigma_Q = avar(\xi) \;  \partial_{\xi}\boldsymbol{Q} \; \partial_{\xi}\boldsymbol{Q} ^* 
\end{equation}

Our three-quantiles estimators $\hat\sigma$ and $\hat\mu$ of   $\sigma$ and $\mu$ are given by
\begin{equation*}
    \hat \sigma = S(\boldsymbol{\hat T}, \boldsymbol{\hat Q}) =  \frac{\hat T_2 - \hat T_1}{\hat Q_2 - \hat Q_1}; \quad \hat \mu = L(\boldsymbol{\hat T},\boldsymbol{\hat Q} ) = \frac{\hat T_1\hat Q_2 - \hat Q_1\hat T_2}{\hat Q_2 - \hat Q_1}
\end{equation*}

Both $S(\boldsymbol{T} , \boldsymbol{Q} )$ and $L(\boldsymbol{T} , \boldsymbol{Q} )$ are smooth functions of $(\boldsymbol{T} , \boldsymbol{Q})$, and  their partial differentials are directly computed as
\begin{align*}
    &\partial_{\boldsymbol{T}} S = \frac{1}{(Q_1-Q_2)} \left[ 1, -1, 0 \right] \quad \partial_{\boldsymbol{Q}} S = \frac{T_1-T_2}{(Q_1-Q_2)^2} \left[ -1, 1, 0 \right] \nonumber\\
    &\partial_{\boldsymbol{T}} L =\frac{1}{(Q_1-Q_2)} \left[ Q_2, - Q_1, 0 \right] \quad \partial_{\boldsymbol{Q}} L = \frac{T_1-T_2}{(Q_1-Q_2)^2} \left[ Q_2, Q_1, 0\right] \nonumber
\end{align*}

We still need to compute the $3 \times 3$  covariance matrix 
$$
Cov_{T,Q} = \lim_{N \to \infty} E(\Delta \boldsymbol{T} \Delta \boldsymbol{Q}^*)
$$
Recall that $\xi = \Phi(\boldsymbol{T})$ and that the line vector $W = \partial_{\boldsymbol{T}} \Phi(\boldsymbol{T}) $ has already been computed above. As is classically  known from Cramer-Wold theorem already used above , in all the following 1st order Taylor approximations of error estimates ,  the 2nd order remainder  terms are irrelevant to compute the gaussian limits of their distributions
Then up to these 2nd order terms we have the 1st order Taylor approximations
\begin{align*}
    &\Delta \xi \approx W \Delta \boldsymbol{T} = W_1 \Delta T_1 + W_2 \Delta T_2 + W_3 \Delta T_3 \nonumber\\
    &\Delta \boldsymbol{Q} \approx \left[ \partial_{\xi}Q_1 , \partial_{\xi}Q_2 ,\partial_{\xi}Q_3 \right] \Delta \xi \nonumber\\
    &\Delta \sigma \approx \partial_{\boldsymbol{T}} S \Delta \boldsymbol{T} + \partial_{\boldsymbol{Q}} S\Delta \boldsymbol{Q} \nonumber\\
    &\Delta \mu \approx \partial_{\boldsymbol{T}} L \Delta \boldsymbol{T} + \partial_{\boldsymbol{Q}} L\Delta \boldsymbol{Q}
     \nonumber
\end{align*}
These relations will now be used to compute the following asymptotic covariances and variances
$$
Cov_{T,Q}= \lim_{N \to \infty} E(\Delta \boldsymbol{T} \Delta \boldsymbol{Q}^*) \\
 avar(\sigma) = \lim_{N \to \infty} var(\Delta\sigma) \; ; \;  avar(\mu) = \lim_{N \to\infty} var(\Delta\mu)
$$
We first compute the coefficients of $A=Cov_{T,Q} =E(\Delta\boldsymbol{T} \Delta \boldsymbol{Q}^*)$. 
$$
A_{i,j} = E(\Delta T_i \Delta Q_j) \approx E(\Delta T_i \partial_{\xi}Q_j \Delta \xi) \approx  
E(\Delta T_i \partial_{\xi}Q_j (W_1 \Delta T_1 + W_2 \Delta T_2 + W_3 \Delta T_3) )
$$
where all approximations derive from the Cramer-Wold theorem. Hence
$$
A_{i,j} \approx \partial_{\xi}Q_j \sum_{k=1, 2, 3} W_k  E(\Delta T_i \Delta T_k ) ) \approx \partial_{\xi}Q_j \sum_{k=1, 2, 3} W_k  \Sigma_T(k,i ) = (W \Sigma_T)_i \partial_{\xi}Q_j 
$$
This yields the  explicit formula 
\begin{equation} \label{covTQ}
Cov_{T,Q}= \partial_{\xi} Q ^* W \Sigma_T 
\end{equation}
Applying again the 1st order approximations listed above yields explicit asymptotic variances for $\hat\sigma$   and $\hat\mu$, namely
\begin{equation} \label{avarsig}
avar(\sigma) = \; \partial_{\boldsymbol{T}}S \; \Sigma_T \; \partial_{\boldsymbol{T}}S ^* + \partial_{\boldsymbol{Q}}S \; \Sigma_Q \; \partial_{\boldsymbol{Q}}S ^*
   + 2 \partial_{\boldsymbol{T}}S \; Cov_{T,Q} \; \partial_{\boldsymbol{Q}}S ^* 
\end{equation}

\begin{equation} \label{avarmu}
avar(\mu) = \; \partial_{\boldsymbol{T}}L \; \Sigma_T \;  \partial_{\boldsymbol{T}}L ^* + \partial_{\boldsymbol{Q}}L \; \Sigma_Q \;  \partial_{\boldsymbol{Q}}L ^*
   + 2 \partial_{\boldsymbol{T}}L \; Cov_{T,Q} \; \partial_{\boldsymbol{Q}}L ^*
\end{equation}

Similar computations also yield explicit expressions $\Sigma_{\theta}$ for the full asymptotic covariance matrix of our estimator $\hat\theta$.
\end{proof}

\begin{corollary}\label{q.opt}
    For an i.i.d. sample $Y_1, \ldots, Y_N$ with CDF $G_\theta$, the three-quantiles estimator $\hat{\xi}_N(\boldsymbol{q})$ of $\xi$ has an asymptotic variance $avar(\xi)$ computed by formula \eqref{eq:AVAR}, which  is a smooth function $AVAR$ of $\boldsymbol{q}$ and $\theta = [\xi, \mu, \sigma]$. This function  has the following natural invariance property:
\begin{equation} \label{invariance}
AVAR((\boldsymbol{q},\xi, \mu, \sigma)) = AVAR(\boldsymbol{q}, \xi, 0, 1)
\end{equation}

For each shape parameter $\xi$, one can then numerically find the optimal choice $0 < q_1(\xi) < q_2(\xi) < q_3(\xi) < 1$ of three percentiles which minimize the asymptotic variance $AVAR(\boldsymbol{q}, \xi, 0, 1)$. The same optimal choice $\boldsymbol{q}(\xi)$ will then also minimize $AVAR(\boldsymbol{q}, \xi, \mu, \sigma)$ for any given pair $\mu, \sigma$.
We have for instance computed and  displayed the optimal choice of three percentiles $\boldsymbol{q}(\xi)$ for $-5 \leq \xi \leq 5$, (see Figure \ref{figure: opt} in Appendix \ref{section: opt3q}).
\end{corollary}

\begin{proof} 
One has $AVAR = W(\boldsymbol{q})\Sigma_T W(\boldsymbol{q})^*$ due to \eqref{eq:AVAR}. Equation \eqref{eq:W(q)} gives $W(\boldsymbol{q}) = aV$ with explicit formulas for the scalar $a$ and the vector $V$. Denote $S_j$ and $T_j$  the $q_j$-quantile of $G_{[\xi,0,1]}$ and $G_{[\xi,\mu,\sigma]}$. Clearly, $T_j = \sigma (S_j - \mu)$. When one replaces $\mu = 0, \sigma = 1$ with arbitrary $\mu$ and $\sigma > 0$, the term $B = (S_3 - S_2)/(S_3 - S_1) = (T_3 - T_2)/(T_3 - T_1)$ is unchanged, and hence the scalar $a$ is unchanged. The vector $V$ is replaced by $\frac{1}{\sigma}V$. The vector $W(\boldsymbol{q})$ thus becomes $\frac{1}{\sigma}W(\boldsymbol{q})$ so that the matrix $\Sigma_T$ given by \eqref{eq:Sigma3} is replaced by $\sigma^2\Sigma_T$. So $AVAR$ remains unchanged.
\end{proof}
The asymptotic variance of $\hat\xi$ is determined by the choice of triplets of percentiles. We display these results further on (see Section \ref{section: opt3q} ).


\section{Multi-Quantile estimators of GEV parameters}\label{section: mq}
As seen above, GEV parameters can be simultaneously estimated by multiple three-quantiles estimators , associated to any finite set $M$  of $m$ percentiles triplets $\boldsymbol{q^s}$. We now  construct  asymptotically normal consistent estimators of the GEV shape parameter $\xi$  having very  high asymptotic accuracy. Our  approach is to compute optimally  weighted linear combinations of  finite sets of three-quantiles estimators of $\xi$.

\subsection{Definition of Multi-Quantile estimators }\label{def:MultiQ}
 
 Fix any set $M$ of $m$ percentiles triplets $\boldsymbol{q}^s$, where $\boldsymbol{q^s} = [q^s_1 < q^s_2 < q^s_3]$. Given $N$ i.i.d. observations sampled from $G_\theta$, each $\boldsymbol{q}^s$ defines a three-quantiles estimator $\eta^s_N = \hat{\xi}_N(\boldsymbol{q^s})$. Fix any vector $w= [w_1, \ldots, w_m]$ of $m$ positive weights adding up to 1.  The linear combination $X_N = \sum_{s=1 \dots m} w_s\eta^s_N$ will be called the  \emph{Multi-Quantile estimator} of $\xi$ defined by the set $M$  and the vector $\mathbf{w}$.

\begin{theorem}\label{asym var multi} 
    (Notations of definition \ref{def:MultiQ}). Fix any set $M$ of $m$ percentiles triplets $\boldsymbol{q^s}$. This  defines  $m$ three-quantiles estimator $\eta^s_N$ of $\xi$. Fix any vector $\mathbf{w}$ of $m$ positive weights.  The Multi-Quantile estimator $X_N$  defined by $M$ and $\mathbf{w}$  is then an  asymptotically normal and consistent estimator of the shape parameter $\xi$ with  asymptotic variance $\tau^2$ given by 
\begin{equation}\label{eq:tau2}
\tau^2 = \mathbf{w} \Lambda \mathbf{w}^*
\end{equation}
    For all $s,t$, the coefficients   $\Lambda(s, t) = \lim_{N\to\infty} N \text{Cov}(\eta^s_N, \eta^t_N)$ are given by:
\begin{equation} \label{eq:covmulti}
\Lambda(s, t) = W(\boldsymbol{q^s})K(\boldsymbol{q^s}, \boldsymbol{q^t})W(\boldsymbol{q^t})^*
\end{equation}
    where each line vector $W(\boldsymbol{q^s}) \in \mathbb{R}^3$ is given by formula \eqref{eq:W(q)}, and the $3 \times 3$ matrices $K(\boldsymbol{q^s}, \boldsymbol{q^t})$ are explicitly computed by
\begin{equation}\label{eq:Kqr}    
K _{i,j}(\boldsymbol{q}^s,\boldsymbol{q}^t)=
\sigma^2\frac{\max(q^s_i,q^t_j)-q^s_i q^t_j}{q^s_i q^t_j(\log(q^s_i)\log(q^t_j))^{1+\xi}}\; \;  \text{for} \; i,j \in [1,2,3]
\end{equation}
    Indeed the vector $\boldsymbol{\eta}_N = [\eta^1_N, \dots, \eta^m_N]$ is an asymptotically normal estimator of the vector $[\xi, \dots, \xi] \in \mathbb{R}^m$, with asymptotic covariance matrix  $\Lambda$.  Note that
    $\Lambda(s, t)$ depends on $\boldsymbol{q^s}$, $\boldsymbol{q^t}$, and  $\xi$, but does not depend on the location and scale parameters $\mu, \sigma$.
\end{theorem} 

\begin{proof}
The triplet of percentiles  $\boldsymbol{q}^s$ defines the vector $\boldsymbol{T(q^s)}$ of true quantiles  for  the GEV $G_{\theta}$. Let $U^s = \boldsymbol{\hat{T}(q^s)}_N$ be the empirical quantiles estimator   of $\boldsymbol{T(q^s)}$ computed from  N observations sampled from $G_{\theta}$. Fix any pair $\boldsymbol{q}^s$ and $\boldsymbol{q}^t$.  As seen in section \ref{section: notation},  the vector $Z_N = [U^s, U^t] \in R^6$   is an asymptotically normal estimator of $[  \boldsymbol{T(q^s)}, \boldsymbol{T(q^t)} ]$. The $3 \times 3$ matrices  $Cov(U^s, U^s)$, $Cov(U^t, U^t)$, $Cov(U^s, U^t)$, and $Cov(U^t, U^s)$ give the natural 4 blocks of the covariance matrix $Cov(Z_N)$.  Formula \ref{eq:Sigma.k} then directly  provides the  four blocks decomposition of the limit $K = \lim_{N\to\infty}Cov(Z_N)$, namely
\begin{align} \label{eq:Kblocks}
K=\begin{pmatrix}
    K(\boldsymbol{q}^s, \boldsymbol{q}^s) && K(\boldsymbol{q}^s, \boldsymbol{q}^t) \\
    K(\boldsymbol{q}^t, \boldsymbol{q}^s) && K(\boldsymbol{q}^t, \boldsymbol{q}^t)
\end{pmatrix}
\end{align}
where, $K(\boldsymbol{q}^s, \boldsymbol{q}^t)$ is given by  \eqref{eq:Kqr}.

For short denote $\eta^s =\eta_N^s$ the three-quantiles estimator of $\xi$ defined by $\boldsymbol{q}^s$. As seen above, there is a smooth function $\Phi(U)$ of $U \in \mathbb{R}^3$ such that $\eta^s = \Phi(U^s)$ and $\xi = \Phi(\boldsymbol{T(q^s)})$. 

Thus $\eta^s, \eta^t$ is a smooth function of $Z_N= [U^s,U^t]$, and hence is an asymptotic jointly normal consistent estimator of $[\xi, \xi]$, with asymptotic covariance matrix $\Delta K \Delta^*$, where $\Delta = [\partial_U\Phi(U), \partial_V\Phi(V)]$ is the differential of $\Phi(U), \Phi(V)$ at the point $U = \boldsymbol{T(q^s)}, V = \boldsymbol{T(q^t)}$. Equation \ref{eq:W(q)}, provides  explicit expressions for   $\partial_U\Phi(U) = W(\boldsymbol{q}^s)$ and $\partial_V\Phi(V) = W(\boldsymbol{q}^t)$, and hence for $\Delta = [W(\boldsymbol{q}^s), W(\boldsymbol{q}^t)]$. By block multiplication, $\Delta K \Delta^*$  directly yields the asymptotic covariance $\Lambda(s,t) = W(\boldsymbol{q}^s)K(\boldsymbol{q}^s, \boldsymbol{q}^t)W(\boldsymbol{q}^t)^*$ of $\eta^s$ and $\eta^t$, as announced in equation \ref{eq:covmulti}.
 
Similarly, the vector of three-quantiles estimators  $\boldsymbol{\eta}_N = [ \eta^1, \ldots, \eta^m ]$ is a smooth function $\Psi(W_N)$ of $W_N = [U^1, \ldots, U^m]$, namely$\Psi(W_N) =[\Phi(U^1), \ldots, \Phi(U^m)]$  Since $W_N$ is a vector of $3 \times m$ empirical quantiles, $W_N$ is an asymptotically normal and consistent estimator of the corresponding vector of true quantiles of $G_{\theta}$. Hence $\boldsymbol{\eta}_N$ must be an asymptotically normal and consistent estimator of $[\xi, \ldots, \xi] \in \mathbb{R}^m$. 

The Multi-Quantile estimator $X_N$ defined by  $ <\; \mathbf{w}, \boldsymbol{\eta}_N \; >$  is a linear function of $\boldsymbol{\eta}_N$, and hence must also be an asymptotically normal and consistent estimator of $\xi$ , with asymptotic covariance $\tau^2$ given by equation \eqref{eq:tau2}. 

Corollary \ref{q.opt} proved that $\Lambda(s, s)$ does not depend on $(\mu, \sigma)$. A similar proof shows that this is  also valid for $\Lambda(s, t)$.
\end{proof}

\begin{definition}\label{optim}(Optimized Multi-Quantile estimators of $\xi$). Same notations as in  theorem \ref{asym var multi}. Fix a set $M = \{ \boldsymbol{q}^1, \dots,\boldsymbol{q}^m\} $ of distinct  percentiles triples. Then $M$ defines a vector $\boldsymbol{\eta}_N = [\eta^1_N, \ldots, \eta^m_N]$ of three-quantiles estimators of the shape parameter $\xi$, and $\boldsymbol{\eta}_N$ has asymptotic covariance matrix $\Lambda$ given by \eqref{eq:covmulti}. For any vector of $m$ positive weights $w$, the Multi-Quantile estimator $< w, \eta_N >$ of $\xi$ has asymptotic variance $\tau^2 = \mathbf{w} \Lambda \mathbf{w}^*$.  Whenever the asymptotic covariance matrix $\Lambda$ is \emph{invertible}, the variance $\tau^2$ is classically minimized by the optimal vector of weights:
\begin{equation}\label{eq:w.opt}
\mathbf{w}_{\text{opt}} = \frac{1}{z^* \Lambda^{-1} z} \Lambda^{-1} z
\end{equation}
where $z \in \mathbb{R}^m$ has all its coordinates equal to 1. This defines the \emph{optimized Multi-Quantile estimator} of $\xi$ as $optX_N = < \mathbf{w_{opt}}, \mathbf{\eta}_N >$.  As seen above, for fixed $\xi$, then $optX_N$ is  an asymptotically normal and consistent estimator of $\xi$, with  asymptotic variance $\tau_{opt}^2$ given by
\begin{equation}\label{eq:tau.opt}
\tau_{\text{opt}}^2 = \frac{1}{z^* \Lambda^{-1} z}
\end{equation}
\end{definition}
The vector $\mathbf w_{\text{opt}}$ of optimal weights, as well as $\tau_{\text{opt}}^2$ do not depend on $(\mu, \sigma)$, but only on $\xi$ and on the set $M$. 

\subsection{Practical computation of  optimized  Multi-Quantile estimators}\label{subsection: practical method}
The  optimal vector of weights $\mathbf w_{\text{opt}}$   cannot be computed directly by \eqref{eq:w.opt} since $\xi$ is naturally unknown but can be consistently estimated by a small number of explicit iterations, as we now outline.  

Fix the set $M$ of $m$ distinct triples of percentiles. This defines  the vector $\boldsymbol{\eta}_N = [\eta^1_N, \ldots, \eta^m_N]$ of $m$ three-quantiles estimators for  $\xi$. The asymptotic covariance matrix $\Lambda$ of  $\boldsymbol{\eta}_N$ is then  a fully explicit  smooth function $L(\xi)$ of $\xi$, given  by  \eqref{eq:covmulti}. Fix the unknown $\xi$ and assume that \emph{the matrix  $\Lambda = L(\xi)$ is invertible} . 

Then  for $r >0$ small enough, $L(y)$ remains invertible whenever $y$ is in the interval $J(r)=(\xi-r,\xi+r)$, and $L(y)^{-1}$ is continuous for $y \in J(r)$. Let $W$ be the set of all vectors $\mathbf{w}$ of positive weights adding up to 1.   Any   Multi-Quantile estimator $X_N(\mathbf{w}) = <\mathbf{w}, \boldsymbol{\eta}_N > $ verifies $|X_N -\xi | \leq \sum_{s=1\ldots m} |\eta^s_N - \xi|$, and  each $\eta^s_N $  tends to $\xi$ in probability. Hence  
 $max_{\mathbf{w} \in W} |X_N(\mathbf{w}) -\xi | $ tends to 0 as $N\to\infty$. So for each small $\epsilon >0$, there is an $N_0$ such that  for each $N >N_0$
$$
P(X_N(\mathbf{w}) \in J(r) \; \text{for all} \; \mathbf{w} \in W) > 1-\epsilon  
$$
Fix $N>N_0$ and N observations sampled from the GEV $G_{\theta}$. Compute the vector of estimators $\boldsymbol{\eta}_N$ which will now  remain fixed. With probability $> (1-\epsilon)$,  the matrices $L^{-1}(X_N(\mathbf{w}) )$ will then be simultaneously  well defined for all $\mathbf{w} \in W$. So  we can  iteratively compute estimates $x(k)$ and $\mathbf{w}(k)$  of $x_i$ and $\mathbf{w}_{\text{opt}}$   by setting $x_0 = <z, \boldsymbol{\eta}_N>$ and

\begin{align*}
\mathbf{w}(k+1) &= \frac{1}{z^* L(x(k))^{-1} z} L(x(k))^{-1} z\\
x_(k+1) &= \; <\mathbf{w}(k) ,\boldsymbol{\eta}_N>
\end{align*}
where as above $z= [1,1,\ldots,1]$.
Fix the number of  iterations at a small value like $k\leq 5$. We have by construction 
$|x(5) -\xi| \leq \max_{s=1 \ldots m} |\eta^s_N -\xi|$. For $N\to \infty$, this upper bound of $|x(5) -\xi|$ tend to 0 , which  implies convergence  of $x(5)$ to $\xi$, and then convergence of $\mathbf{w}(5)$ to $\mathbf{w}_{\text{opt}}$.

\section{Multi-Quantile estimators with high  number of percentiles triples}\label{section: efficient}
\subsection{Robust sets of percentiles triples} \label{robust}
 A set $M$ of  $m$ distinct percentiles triples will be called \emph{robust} if the $m \times m$ covariance matrix $\Lambda= L(\xi) $ is indeed invertible for all 
values of the scale parameter $\xi$.\\
For each $m$, a natural question  is how to select "robust" sets $M$ of $m$ distinct percentiles triples. We have  rigorously analyzed this point  for small values of $m$, as reported in the Appendix. We have not pursued a full mathematical answer for generic large $m$, because we have  numerically  validated that the following   randomized  selection of $M$  yields robust sets $M$ with high probability.

\subsection{Randomized choice of robust sets of percentiles triples}\label{random M}
Fix any integer $m$ and  let $r = m+2$. Denote $E(r)$ the set of equidistant percentiles 
$j/r$ with $j=1,\ldots, (r-1)$. Randomly select  a set $M(m)$ of $m$  distinct  percentiles triples from  the set $E(r)$. We conjecture that for each fixed $m$ this randomly selected set $M(m)$ will be "robust" with very high probability $p(m)$.

We have numerically checked this conjecture for $m \leq 100$, and our estimates indicate that $p(m)$  increases with $m$. 

\subsection{Accuracy of Multi-Quantile estimators based on high number of quantiles}
To evaluate how large $m$ values impact the variance of optimized Multi-Quantile estimators of $\xi$, we have implemented the following numerical tests.

\begin{enumerate}
    \item Fix the shape parameter $\xi$.
    \item Successively increase  the number $m$ of percentiles triples from  $ 10$ to $100$
    \item For each fixed pair $\xi, m$ , randomly select a set $M(m)$ of $m$  distinct  percentiles triples within the set $E(r)= \{ 1/r, 2/r, \dots, (r-1)/r\}$ of $r= m+2$ equidistant percentiles. Check numerically if $M(m)$ is robust. If this is not the case, randomly select another set $M(m)$. Robustness of $M(m)$ is typically reached after very few such attempts.
    \item The set $M(m)$ defines an optimized Multi-Quantile estimator $X_{opt}(m)$ of $\xi$ as outlined in definition \ref{optim}  
    \item Using formulas \eqref{eq:covmulti} , \eqref{eq:w.opt}, \eqref{eq:tau.opt} compute the asymptotic variance $\tau^2_{opt}(m)$ of our optimized Multi-Quantile estimator $optX(m)$ 
\end{enumerate}

Note that these numerical computations are only based on our explicit theoretical expression for $\tau^2_{opt}(m)$, and hence do not at all involve simulating large size numbers $N$ of observations sampled from $G_{\theta}$. 

Typical results of this numerical study are displayed in Figure \ref{figure: diminsh error}. For fixed  $\xi\geq -0.5$, as  $m$ increases from 10 to 100, the asymptotic variance $\tau^2_{opt}(m)$ of our optimized Multi-Quantile estimator $optX(m)$ decreases rapidly to the Cramer-Rao lower bound $CRB(\xi)$. When $\xi < -0.5$, the lower bound 
$CRB(\xi)$ no longer exists, but $\tau^2_{opt}(m)$  still decreases as $m$ increases, to nearly reach its lowest value denoted $\tau^2_{opt}(\infty)$ .

\begin{figure}[!h]
     \centering
     \begin{subfigure}[b]{0.47\textwidth}
         \centering
         \includegraphics[width=\textwidth]{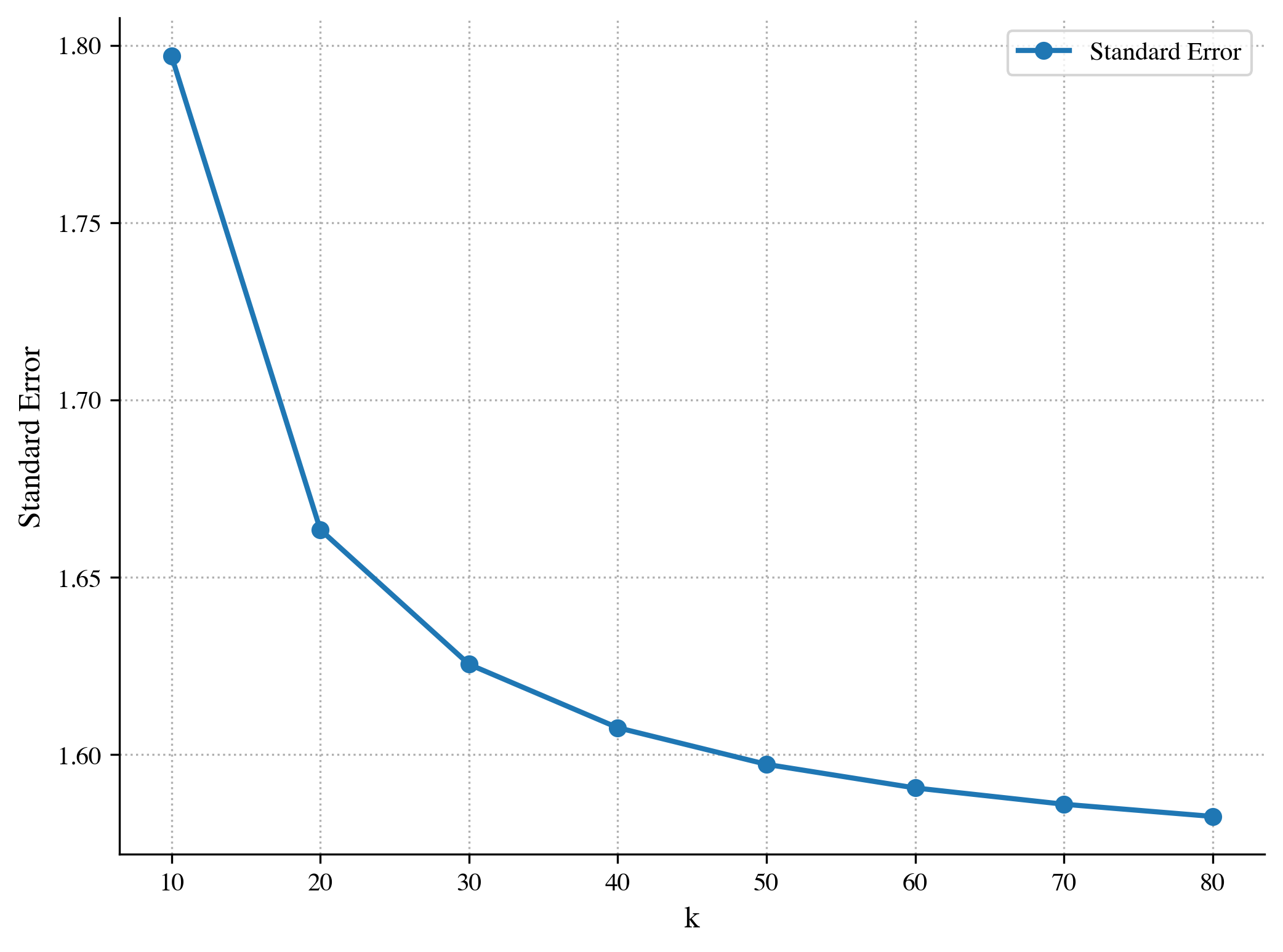}
         \caption{$\xi =-2$}
         \label{xi n2}
     \end{subfigure}     
     \hfill
     \begin{subfigure}[b]{0.47\textwidth}
       \centering
        \includegraphics[width=\textwidth]{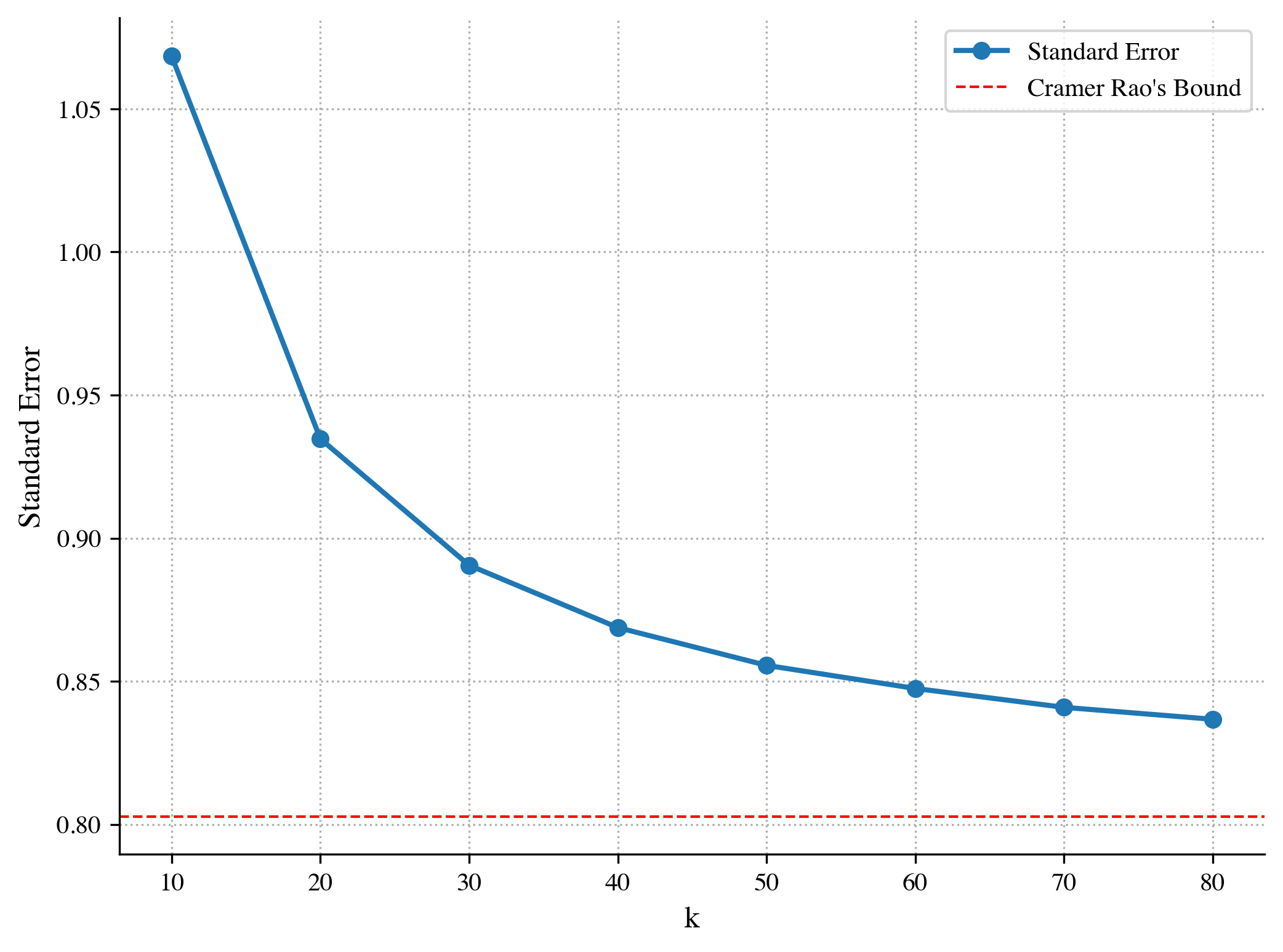}
         \caption{$\xi =0.2$}
        \label{xi p02}
      \end{subfigure}     
     
        \caption{For fixed shape parameter $xi$ and increasing number $m$ or percentage triples, note the decreasing asymptotic variance $\tau^2_{opt}(m)$ of optimized Multi-Quantile estimators $optX(m)$ of $\xi$ based on $m$ distinct percentiles triples.  For $\xi>-0.5$ the Cramer-Rao lower bound $CRB(\xi)$ on variances for  estimators of $\xi$ does exist, and as $m$ increases, $\tau^2_{opt}(m)$ decreases towards $CRB(\xi)$. For instance in figure (b) where $\xi= 0.2$, the value  $CRB(0.2) =0.805$ (red line) is nearly reached at $m=80$ with $\tau^2_{opt}(80) = 0.846$.  For $\xi <-0.5$ the Cramer-Rao bound does not exist, but $\tau^2_{opt}(m)$ still decreases to a lower limit $\tau^2_{opt}(\infty)$ as $m$ increases. For instance when $\xi= -2$ (figure (a) ) , $\tau^2_{opt}(m)$ decreases to $\tau^2_{opt}(80) = 0.76 $ .}
        \label{figure: diminsh error}
\end{figure}



\section{Comparison between our Multi-Quantile estimators and  classical estimators}\label{section: compare}
\subsection{Classical estimators of GEV parameters} 
Let $G_{\theta}$ be a GEV distribution with unknown $\theta= [ \xi,\mu,\sigma]$. Let $Y_1,...,Y_N$ be an  i.i.d random sample from $G_{\theta}$. To estimate $\theta$ from $Y_1,...,Y_N$, several  estimators have been proposed in the literature, and we have compared performances between our Multi-Quantile estimators and three other classical estimators: the Maximum Likelihood Estimator (MLE), the PWM estimator (Hosking, 1985), the DEH estimator (Dekkers, Einmahl,de Haan,  1989) 

\subsubsection{MLE estimators} This estimator  $\hat\theta_{MLE}$ aims to maximize in $\theta$ the log-likelihood function
 \begin{equation*}
        L(\theta)=\sum^N_{i=1} -\log\sigma+(1+1/\xi )\log(1+\xi \frac{y_i-\mu}{\sigma})-(1+\xi  \frac{y_i-\mu}{\sigma} )^{-1/\xi }
    \end{equation*}
by numerically solving the non-linear system of three equations $grad_{\theta} L = 0$ under the natural constraints 
\begin{equation*}
1 + \xi (Y_j-\mu)/\sigma > 0 \; ; \text{for all} \; j=1 \ldots N\}
\end{equation*}
When $\xi > - 0.5$, MLE estimators are asymptotically normal and consistent estimators of $\theta$ standard convergence rate $1/\sqrt{N}$  and their asymptotic covariance matrices can  be numerically evaluated (see Appendix \ref {appendix: asym xi}) However when $-1 < \xi < - 0.5$, the MLE estimators are not asymptotically normal, and when $\xi < -1$, the maximum of the likelihood function is  not even  well defined so that  consistent MLE estimators are not available .

\subsubsection{PWM estimators} Re-ordering the sample $Y_1,...,Y_N$ yields the vector  $Z_1 <  Z_2 < \ldots < Z_N$ of order statistics.  The  PWM estimator $\hat\xi_{PWM}$ of $\xi$ is defined only for $\xi <0.5$ and is then computed as the solution $x$ of
\begin{equation*}
    \frac{3^x-1}{2^x-1} = \frac{3\beta_2-\beta_0}{2\beta_1-\beta_0}
\end{equation*}
where
\begin{align*}
\beta_0 &= \frac{1}{N}\sum_{i=1}^N  Z_i \\ 
\beta_1 &= \frac{1}{N}\sum_{i=1}^N\frac{i-1}{N-1} Z_i \\
\beta_2 &= \frac{1}{N}\sum_{i=1}^N\frac{(i-1)(i-2)}{(N-1)(N-2)}Z_i
\end{align*}
For $\xi < 0.5$ the PWM  estimator is asymptotically normal with standard convergence rate $1/\sqrt{N}$ and has a numerically computable asymptotic variance (see Appendix \ref{appendix: asym xi} ). 

\subsubsection{DEH estimators} Given the sample $ Y_1,...,Y_N$,  the DEH estimator focuses only on the top $k= k(N) < N$  order statistics $Z_{N-k+1}, \ldots, Z_N$ of the sample. For asymptotic consistency the integers $k(N)$ are constrained by $\lim_{N\to\infty}k(N)=\infty$ and $\lim_{N\to\infty}\frac{k(N)}{N}=0$. Once $k(N)$ has been  selected, the DEH estimator of $\xi$ is computed for each $N$ by 

\begin{align*}
    H_1 &= \frac{1}{k}\sum^{k-1}_{i=0}   \log\frac{Z_{N-i}}{Z_{N-k}}\\
    H_2 &= \frac{1}{k}\sum^{k-1}_{i=0}   (\log\frac{Z_{N-i}}{Z_{N-k}})^2 \\
    \hat\xi_{DEH} &= H_1  +2\frac{H_1^2}{H_2} -1
\end{align*}
where $k= k(N)$. The DEH estimator  is  asymptotically normal without restrictions on $\xi$, but has a convergence rate $1/\sqrt{k(N)}$, which is much slower than the convergence rate $1/\sqrt{N}$ of all the other estimators considered here. Moreover efficient choice of $k(N)$ is strongly dependent on the unknown $\xi$ , a substantial disadvantage as will be clear below.
\subsection{Comparison of empirical estimation errors by finite sample simulation }
To evaluate the finite-sample behavior of MQ, MLE, PWM, and DEH estimators , we implemented an extensive empirical study by  Monte Carlo simulations of random samples from a GEV distribution  with fixed shape parameter $\xi$. Each sample had size $N=1000$, and $K=1000$ such samples were simulated  for each $\xi$ in the set $-3, -2, -1 , -0.2, 0, 0.2, 1,2$. 

For each $\xi$, we applied the MQ, MLE, PWM, DHE estimators to K simulated samples of size N . Each estimator  MQ, MLE, PWM, DHE  thus generated  $K$ estimates of $xi$, and hence provided the  empirical bias as well as the empirical standard errors of estimation $err_N(MQ)$, $err_N(MLE)$, $err_N(PWM)$, and $err_N(DHE)$. These accuracy results (bias and standard errors) for samples of size $N=1000$ are presented graphically in Fig. \ref{figure: empirical error}, and  Table \ref{table: estimation error} displays $err_N(MQ)$, $err_N(MLE)$, $err_N(PWM)$ and $err_N(DHE)$. These empirical results for moderate sample size $N= 1000$ validate concretely the theoretical standard errors presented in Table \ref{table: asym var}, which a priori rely on  much larger values of $N$. 

Our Multi-Quantile estimator MQ  maintains low $err_N(MQ)$ and almost no bias over the whole range of $\xi$ values. Moreover MQ  is available for any $\xi$ value, while both MLE and PWM require strong restrictions of $\xi$ values, namely $\xi > -0.5$ for MLE asymptotic normality, and $\xi<0.5$ for PWM availability. Within its restricted availability range,  MLE has the smallest standard error $err_N(MLE)$, but when $|\xi| \geq 0.5$ such as  $\xi = -3, -2, -1,  2$, our MQ estimator clearly stands out as the most accurate estimator.

For the  $DHE$ estimator, the number $k(N)$ of order statistics defining $DHE$  was uniformly set at $k(N)=100$. But the best choice for $k(N)$ requires unavailable previous knowledge of $\xi$, and hence some yet undefined iterative version of $DHE$. So we mitigated our uniform choice of $k(N)$ by increasing $N$ to $N=10,000$. However even this strongly increased sample size did not rescue the weak accuracy of $DHE$,  which for each $\xi$ value clearly had the worst  $err_N(DHE)$ among the 4 estimators tested here . 


\subsection{Comparison of asymptotic variances}
Detailed formulas for the asymptotic variances of  PWM estimators, and  for MLE estimators when $\xi >0.5$ are given in  the appendix (see Appendix \ref{appendix: asym xi}). Since  DEH estimators of $\xi$ have convergence rate $1/\sqrt{k(N)}$ much slower than the standard $1/\sqrt{N}$, the DEH  asymptotic variances are not concretely relevant here , so we omit $DHE$ in our presentation of asymptotic results.

For our Multi-Quantile estimator MQ, we fix a set $E(100)$ of $99$ equidistant percentiles $0.01=q_1 < q_2 < \ldots < q_{99}=0.99$. We randomly select and then fix a robust set $M(98)$ of $m= 98$ distinct  percentiles triples. Using formula \eqref{eq:tau.opt} we compute the asymptotic variance $\tau^2_{opt}$ for the optimized Multi-Quantile estimator of $\xi$ defined by $M(98)$. We did note that for fixed $\xi$ the numerical value of $\tau^2_{opt}$ did not vary significantly whenever we did randomly select another robust $M(98)$.

Hence the three asymptotic variances $var_{MQ}, var_{MLE}, var_{PWM}$ of $MQ, MLE,PWM$ estimators of $\xi$ only  depend on  $\xi$, and were computed numerically from explicit formulas. Results are displayed in Fig. \ref{figure: asymptotic variance} : red curve for $var_{MQ}$, black dashed line for $var_{MLE}$, blue dashed line for $var_{PWM}$. For $\xi >-0.5$, one  has $var_{MLE} < var_{MQ}$ but $var_{MQ}$ remains remarkably close to $var_{MLE}$. For $\xi < -0.5$, the asymptotic variance $var_{MLE}$ does not exist any more but  $var_{MQ}$ always exists and remains significantly smaller than  $var_{PWM}$ with  a ratio $var_{PWM}/ var_{MQ} >> 1$  which keeps increasing as  $\xi$ decreases beyond $-1$. Recall also that  PWM itself does not exist for $\xi >0.5$

For a hypothetical set of $N=1000$ i.i.d. observations sampled from $G_{\theta}$ , we also compute the theoretical standard errors $\sqrt{N var_{MQ}}, \sqrt{N var_{MLE}}, \sqrt{N var_{PWM}}$ for estimators of $\xi$. We display these standard errors in Table \ref{table: asym var} for more concrete  comparisons with the empirical errors of estimation evaluated above for actually simulated samples of size N=1000.
\begin{figure}[!h]
    \centering
    \includegraphics[width=\textwidth]{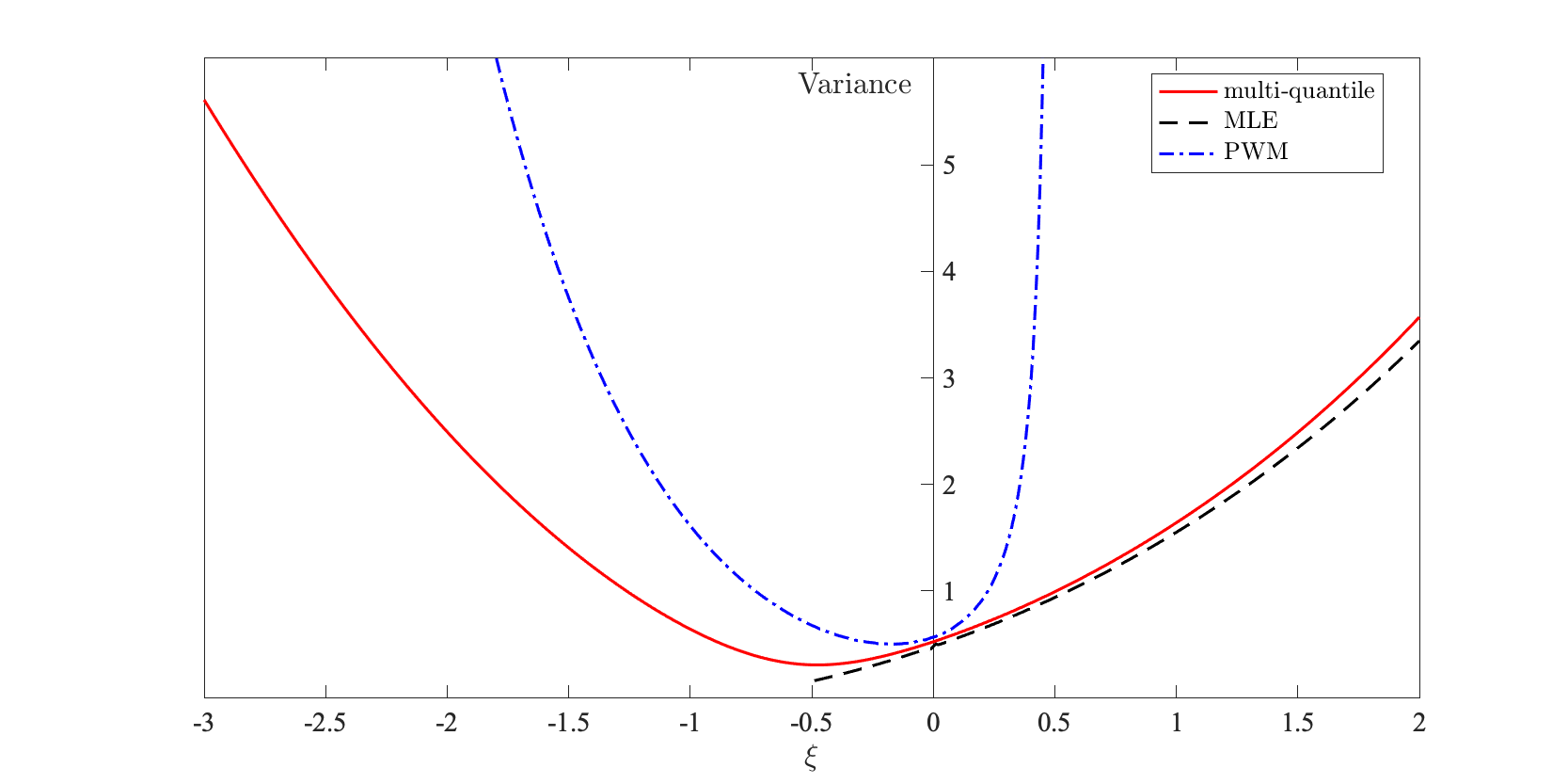}
    \caption{Asymptotic variances $var_{MQ}, var_{MLE}, var_{PWM}$ for three  estimators of shape parameter $\xi$. Solid line for $var_{MQ}$, red dashed line for $var_{MLE}$,  black dashed line for $var_{PWM}$. No formula is available for $var_{MLE}$ when $-1 < \xi <-0.5$ , and for $\xi<- 1$, the  standard MLE estimator  itself is not reliably computable.  For $\xi>0.5$, the estimator PWM  does not exist.} 
    \label{figure: asymptotic variance}
\end{figure}

\begin{table}[!h]
\centering
\caption{For each $\xi=-3,-2,-1,-0.2,0,0.2,1,2$ we have actually simulated $K=1000$ random  samples of size  $N=1000$, with CDF $G_{\theta}$. We have then applied the four estimators $MQ,MLE, PWM, DHE$ of $\xi$ to all these random samples to evaluate their empirical standard errors of estimation. The table displays the good accuracy of MQ over the whole range of $\xi$ values, and outperforms all other estimators for $|\xi| > 0.5$. }
\label{table: estimation error}
\resizebox{0.7\textwidth}{!}{%
\begin{tabular}{l|cccccccc|}
\cline{2-9}
                                     & \multicolumn{8}{c|}{$\xi$}                                            \\ \hline
\multicolumn{1}{|l|}{Estimator}         & -3     & -2     & -1     & -0.2   & 0      & 0.2    & 1      & 2      \\ \hline
\multicolumn{1}{|l|}{Multi-Quantile} & 0.094 & 0.062 & 0.039 & 0.030 & 0.033 & 0.036 & 0.053 & 0.082 \\ \hline
\multicolumn{1}{|l|}{MLE}            & NaN    & NaN    & NaN    & 0.019 & 0.023 & 0.026 & 0.040 & 0.729 \\ \hline
\multicolumn{1}{|l|}{PWM}            & 0.173 & 0.085 & 0.037 & 0.022 & 0.025 & 0.032 & NaN    & NaN    \\ \hline
\multicolumn{1}{|l|}{DHE}            & 0.676 & 0.416 & 0.214 & 0.107 & 0.100 & 0.101 & 0.138  & 0.222  \\ \hline
\end{tabular}%
}
\end{table}

\begin{table}[!h]
\centering
\caption{ For each  $\xi=-3,-2,-1,-0.2,0,0.2,1,2$, and for three estimators of $\xi$, namely our Multi-Quantile MQ, as well as MLE and PWM, we compute standard errors of estimation for a  nominal sample size  $N=1000$ by $\sqrt{N var_{MQ}}, \sqrt{N var_{MLE}},\sqrt{N var_{PWM}}$. Here the asymptotic variances $var_{MQ}, var_{MLE},var_{PWM}$ are derived from  theoretical formulas. Hence no simulations of random samples was necessary for this table. To explain the NaN values in this table, recall that asymptotic variances for MLE are not available when $\xi <- 0.5$. Indeed consistent MLE is not even reliably computable for $\xi <-1$. Similarly, PWM does not exist for $\xi > 0.5$. }
\label{table: asym var}
\resizebox{0.7\columnwidth}{!}{%
\begin{tabular}{l|cccccccc|}
\cline{2-9}
                                     & \multicolumn{8}{c|}{$\xi$}                                                                                                                                                                                                      \\ \hline
\multicolumn{1}{|l|}{Estimator}         & -3    & -2    & -1    & -0.2  & 0   & 0.2   & 1     & 2 \\ \hline
\multicolumn{1}{|l|}{Multi-Quantile} &  0.075 &    0.050&    0.025&    0.020&       0.023  &  0.026 &    0.041 &   0.060              \\ \hline
\multicolumn{1}{|l|}{MLE}            & NaN   & NaN   & NaN   & 0.018 & 0.021 & 0.025 & 0.039 & 0.058                  \\ \hline
\multicolumn{1}{|l|}{PWM}            & 0.185&    0.090&    0.040&    0.022&       0.024 &   0.030 & NaN   & NaN                    \\ \hline
\end{tabular}%
}
\end{table}

\begin{figure}[!h]
     \centering
     \begin{subfigure}[b]{0.7\textwidth}
         \centering
         \includegraphics[width=\textwidth]{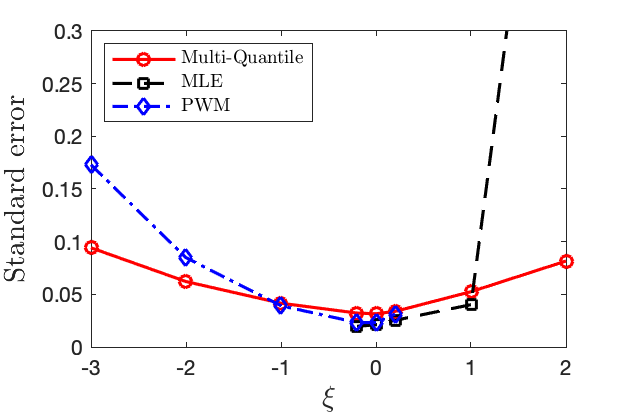}
         \caption{standard error}
         \label{std error}
     \end{subfigure}     
        \caption{Plot of empirical bias and standard error in the  estimation of $\xi$. The three plots on each image correspond to the estimators  MLE, PWM, and our Multi-Quantile MQ. For each $\xi=-3,-2,-1,-0.2,0,0.2,1,2$, these empirical results are computed by simulating $K=1000$ random sammples of size $N=1000$.}
        \label{figure: empirical error}
\end{figure}

\subsection{Comparison of computation times between  MLE and Multi-Quantile  estimators}
We have compared (on a standars laptop) the computation times of MLE and our MQ estimators for simulated samples having a GEV distribution with \(\xi = 0.2\). Sample sizes $N$ ranged from $1000$ to $10,000$. Figure \ref{figure: computation} displays the median CPU times $med_{MQ}$ and $med_{MLE}$, with error bars at  the $5\%$ and $95\%$ quantiles. Clearly  one always has $med_{MQ} < med_{MLE}$ and the ratio $med_{MLE}/med_{MQ}$ increases steadily from $1.3$ to $1.5$ as sample size $N$ increased from $1000$ to $10,000$.

\begin{figure}[!h]
     \centering
     \begin{subfigure}[b]{0.7\textwidth}
         \centering
         \includegraphics[width=\textwidth]{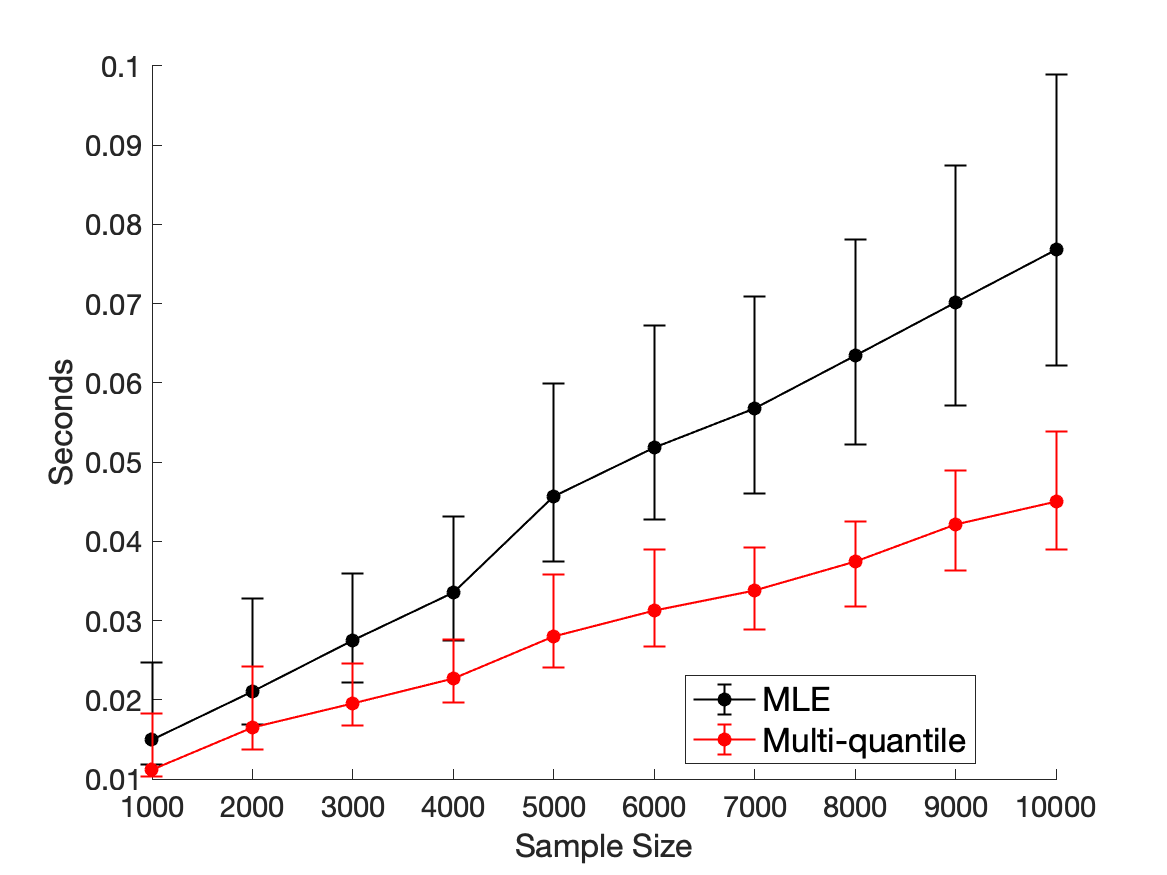}
         \label{fig: comp_sec}
     \end{subfigure}     
     %
        \caption{Comparison of computation times between Multi-Quantile MQ and MLE estimators of the shape parameter $\xi$. The graph displays median CPU times $med_{MQ}$ and $med_{MLE}$ for sample sizes $N$ increasing from $1000$ to $10,000$. Clerly,  $med_{MQ}$ is significantly smaller than $med_{MLE}$, and the ratio $med_{MLE}/med_{MQ}$ increases steadily from $1.3$ to $1.5$ when $N$ increases. }
        \label{figure: computation}
\end{figure}

\section{Extension to Block Maxima estimators}\label{section: block maxima}
We will now extend our Multi-Quantile estimators to the more general Block Maxima framework developed in e.g. \cite{ferreira2015} and \cite{dombry2019}. 

\subsection{Domain of attraction}

Write for short $G(\xi)$ instead of $G_{[\xi,0,1]}$. Let $Z_1, \ldots, Z_N$ be an i.i.d sequence of random variables having the same CDF $F(z)$.  Denote $U_N= \max\left({Z_1, \ldots, Z_N}\right)$, and as in the block maxima literature, we will assume that \emph{ $F$ belongs to the max-domain of attraction $D(\xi)$ } of the GEV $G(\xi)$. This means that there are two sequences $a_N>0$ and $b_N$ such that, the recentered and rescaled variables $\frac{U_N-b_N}{a_N}$ converge in distribution to $G(\xi)$ as $N\to \infty$.

The domain of attraction $D(\xi)$ has been completely characterized by fairly technical conditions involving the rate of decay for the  tail of $F$ (see the reference book \cite{dehaan2006}). However simpler sufficient conditions can define more restricted but still practically useful subsets of $D(\xi)$. For instance, let $S$ be the support of $F$, and let $z^*= \max(S)$, which is not required to be finite. Assume that $S$ contains a non-empty interval $J= (a, z^*)$ with $a$ finite, such that for $z\in J$, $F$ has two continuous derivatives $F'$ and $F''$, with $F'(z)>0$ and $\frac{1-F(z)}{F'(z)}$ tends to some finite $\xi$ as $z\to z^*$. Then $F$ must belong to the domain of attraction of $D(\xi)$, denoted as $F\in D(G_\xi),\ \xi\in\mathbb R$. As standard examples, the Gaussian and exponential distributions belong to $D(0)$, and uniform distributions on finite intervals belong to $D(-1)$. 

\subsection{Block Maxima algorithm}

Block maxima algorithm aim to construct consistent estimators of $\xi$ based on the $N$ i.i.d. observations $Z_j$ having the same CDF $F(z)$.  Take $N= n m$, with $n$ and $m$ are large enough. Divide the set of N observations into $n$ consecutive blocks $B_1, ..., B_n$ of identical size $m$.  Define the $n$ \emph{block maxima} $Y_1, \ldots, Y_n$ by  $Y_i = $ maximum of all $Z_j$ in block $B_i$. 

Assume that $F\in D(\xi)$, and let $a_m, b_m$ be the associated recentering and rescaling constants associated to $F$. Fix any percentile $0 < q <1$. Denote $\hat{T}_{n,m}(q)$ the empirical quantile of $Y_1, \ldots, Y_n$ associated to $q$. As shown by theorem B.3.1 in \cite{dehaan2006}, when both the block size $m$ and the number of blocks $n$ tend to $\infty$, then normalized empirical quantile $H_{m,n}(q) = (\hat{T}_{n,m}(q) - b_m)/a_m$ tends to the true q-quantile $T(q)$ of $G(\xi)$.

Another technical theorem proved in de Haan (2015) and Alex(2019) introduces more constraints on the CDF $F$, to guarantee the asymptotic normality of $\sqrt{m} \left[(H_{m,n}(q) -T(q)\right]$. Fix $r$ percentiles $q_1 < \ldots < q_r$. The same two papers provide theoretical formulas for the asymptotic covariance matrix of the vector of empirical Multi-Quantile $\frac{1}{m} [\;H_{m,n}(q_1), \ldots,H_{m,n}(q_r)$
These formulas are however unwieldy to apply for practical computations since they depend on very precise properties of the unknown distribution $F$. Moreover the rescaling and recentering sequences $a_m, b_m$ are also unknown and need to be estimated from data. 

But the existing asymptotic results just quoted above indicate that the block maxima approach can enable the application of our Multi-Quantile estimators even when $F$ is unknown, provided one relies on empirical estimations of $a_m, b_m$ and on intensive simulations to derive empirical estimates of the relevant covariance matrices. Indeed we have  successfully tested the block maxima extension of our Multi-Quantile estimators for risk analysis on large sets of intraday stock prices data (see our companion paper \cite{Lin2024}).

%
%

\begin{funding}
This work was supported by the National Natural Science Foundation of China (72301126)
\end{funding}


\section{APPENDIX (A) : Optimal choice of percentiles for three-quantiles estimators}\label{section: opt3q}

As seen in Theorem \ref{th:3q normality}, the asymptotic variances of the three-quantiles estimators $ \hat\mu_N(\boldsymbol{q})$ and $ \hat\sigma_N(\boldsymbol{q})$ are essentially determined by the asymptotic variance of $\hat\xi_N(\boldsymbol{q})$. So we have first implemented a numerical computation of the optimal triplet $\boldsymbol{q} = [q_1(\xi) < q_2(\xi) < q_3(\xi)]$ minimizing the asymptotic variance of $\hat\xi_N(\boldsymbol{q})$. These results are displayed in Fig. \ref{figure: opt per} for $-5 \leq \xi \leq 5$, where one always has  $q_1(\xi) \leq 0.037 $, $0.027 \leq q_2(\xi) \leq 0.832$, and $q_3(\xi) \geq 0.827$. When $\xi$ increases from $-0.5$ to $4$, then $q_2(\xi)$ drops drastically from $0.832$ to $0.027$, and keeps decreasing when $\xi$ increases from $4$ to $5$. Interestingly, when $\xi\leq-0.5$, the optimal triplet is fixed at $q_1 = 0.037$, $q_2 = 0.832$ and $q_3 = 0.987$.

Denote \( CRB(\xi) \)  the classical Cramer-Rao lower bound for the asymptotic variance of any consistent estimator of \( \xi \). For \( \xi \leq -0.5 \), the Fisher information matrix of \( G_\theta \) involves a divergent integral, and \( CRB(\xi) \) is not well-defined.

However, for \( \xi > -0.5 \), one can numerically compute \( CRB(\xi) \) and compare it to the asymptotic variance \( avar(\xi) = AVAR(\boldsymbol{q}(\xi)) \) of the optimized Three-Quantiles estimator of \( \xi \). We have displayed our numerical comparisons in \ref{figure: opt var} for \( -0.5 < \xi \leq 5 \). The efficiency ratio \( CRB(\xi)/avar(\xi) \) is, of course, less than 1 but increases from {0.730 to 0.827} as \( \xi \) increases from \(-0.5\) to \(2\), and then decreases slightly to \( 0.817 \) as \( \xi \) continues to increase.

\begin{figure}[!h]
     \centering
     \begin{subfigure}[b]{0.48\textwidth}
         \centering
         \includegraphics[width=\textwidth]{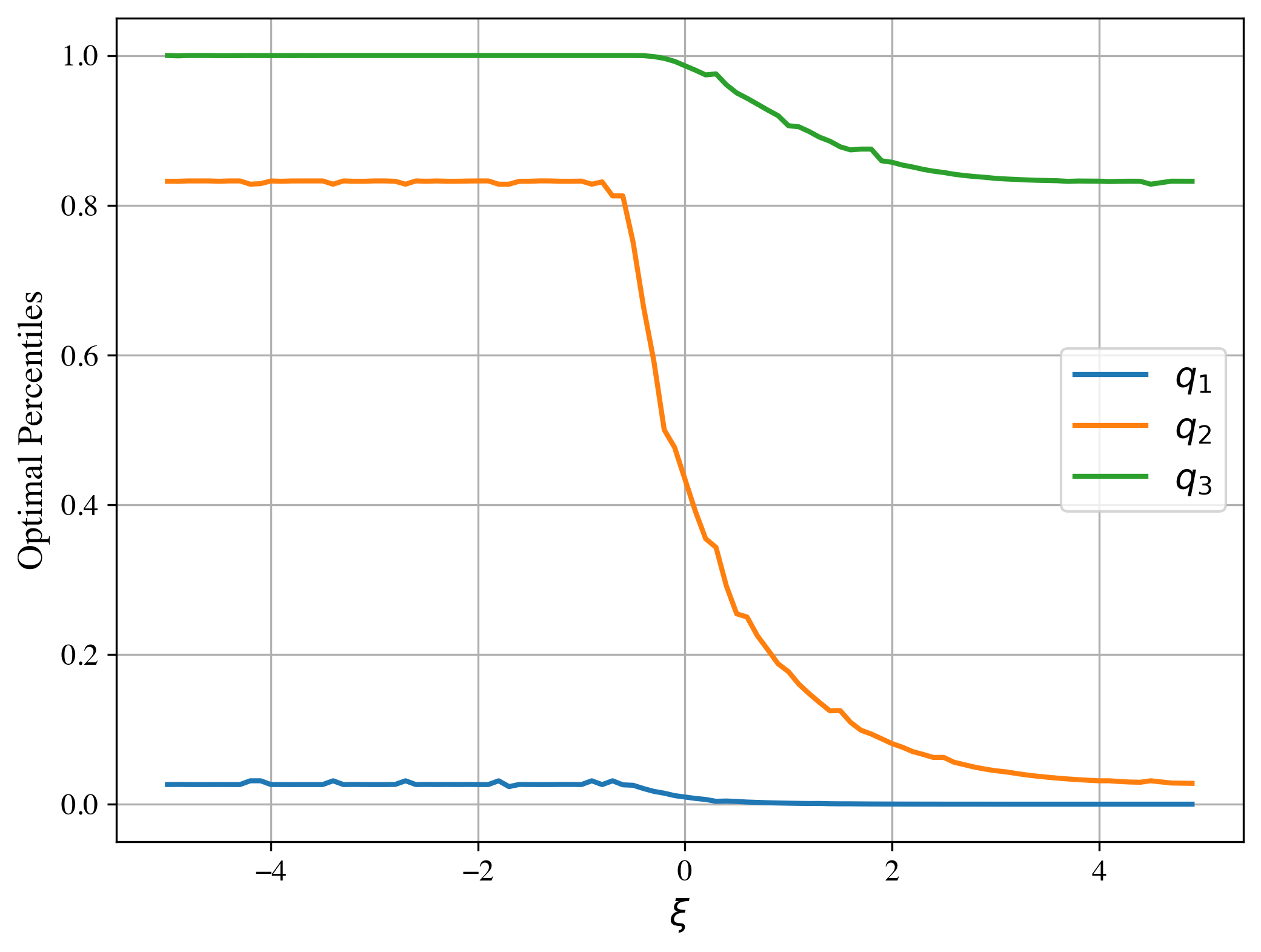}
         \caption{Optimal triplet of percentiles}
         \label{figure: opt per}
     \end{subfigure}     
     \hfill
     \begin{subfigure}[b]{0.48\textwidth}
         \centering
         \includegraphics[width=\textwidth]{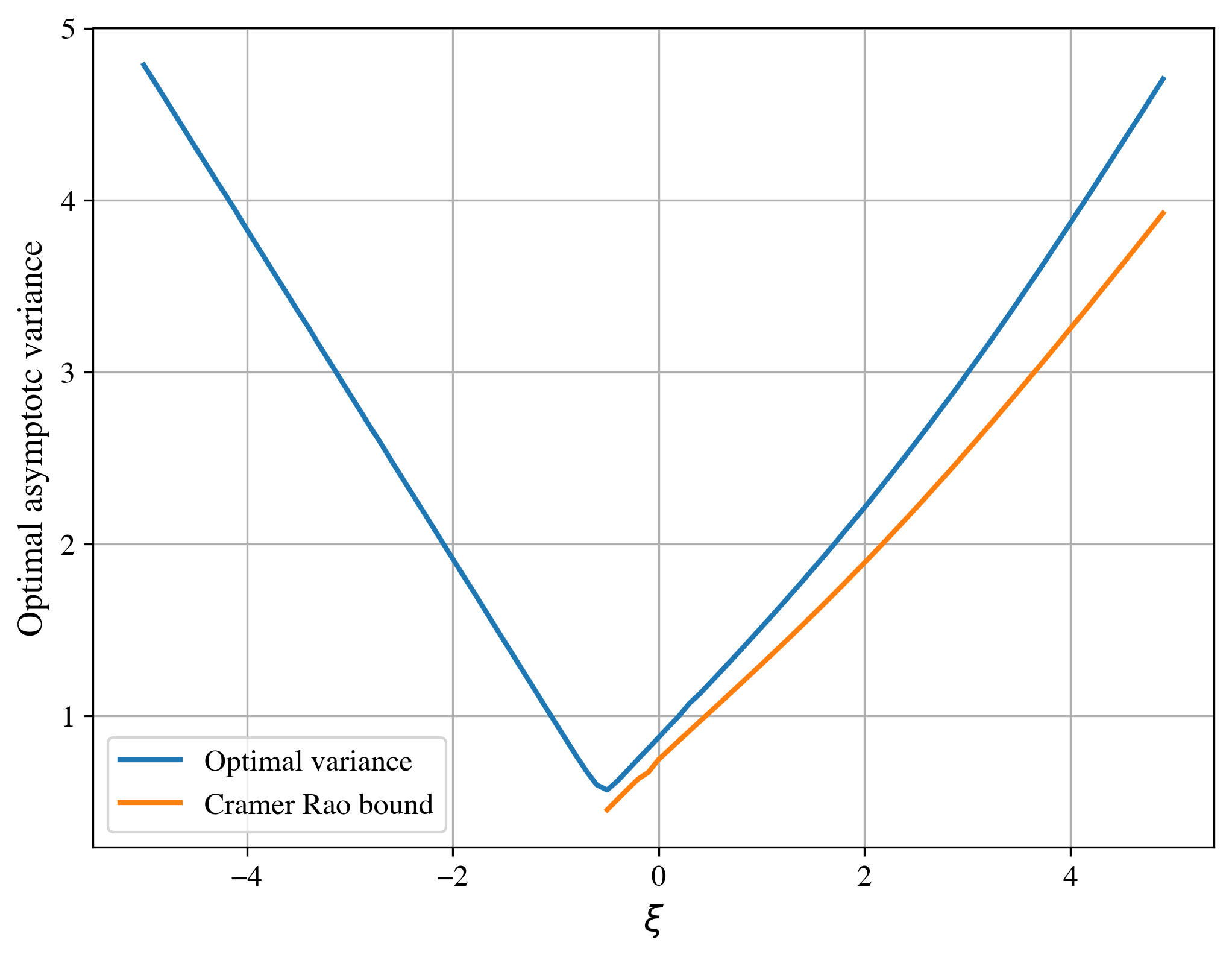}
         \caption{Optimal variance}
         \label{figure: opt var}
     \end{subfigure}     
     
        \caption{For $-5 \leq \xi \leq 5$ the left image displays the three optimal percentiles $\{ 0 < q_1(\xi) < q_2(\xi) < q_3(\xi) <1 \}$ minimizing the asymptotic variance $avar(\xi)$ of $\hat{\xi}_N$. This minimal asymptotic variance $avar^*(\xi)$ is displayed by the right image. For $-0.5 \leq \xi$ , note that $avar^*(\xi)$ is reasonably close to the Cramer-Rao lower  bound $CRB(\xi)$,  graphed in orange, and that both $avar^*(\xi)$ and $CRB(\xi)$ increase when $ \xi $ increases. For $\xi \leq -0.5$  but our optimal three quantiles remain nearly constant,  and $avar^*(\xi) /|\xi|$ remains at moderate levels. }
        \label{figure: opt}
\end{figure}

Our optimal three-quantiles estimators are easily computed numerically, but the optimal percentiles triple \(\boldsymbol{q(\xi)}\) depends on the unknown \(\xi\). This point  naturally led us to introduce the  much more efficient optimized Multi-Quantile estimators studied in this paper

\section{APPENDIX (B): Asymptotic variances for MLE, PWM and DEH estimators}\label{appendix: asym xi}
This section Asymptotic variance for MLE exists only when $\xi >-1/2$, PWM when $\xi<1/2$, and DEH when $\xi\in\mathbb{R}$.

\textbf{MLE}: When $\xi >-1/2$, \cite{bucher2017} proved the asymptotic normality of $\hat\theta_N^{(MLE)}$ namely 
\begin{equation}
    \sqrt{N}(\hat\theta_N^{(MLE)}-\theta)\xrightarrow{d}\mathcal{N}(0,J^{-1}_{\theta})
\end{equation}
where $J_\theta$ is the information matrix of $G_{\theta}$, given by 
\begin{equation}\label{equation: fisher}
    J_{\theta} = \begin{pmatrix}
        \frac{1}{\xi^2}(\frac{\pi^2}{6}+(1-\gamma +\frac{1}{\xi})^2-\frac{2q}{\xi}+\frac{p}{\xi^2}) \; & & \;  -\frac{1}{\xi}(q-\frac{p}{\xi}) \; & & \; -\frac{1}{\xi^2}(1-\gamma -q+\frac{1-r+q}{\xi}) \\
        -\frac{1}{\xi}(q-\frac{p}{\xi}) & & p & & -\frac{p-r}{\xi} \\
        -\frac{1}{\xi^2}(1-\gamma -q+\frac{1-r+q}{\xi}) \;  & & -\frac{p-r}{\xi} & & \frac{1}{\xi^2}(1-2r+p)
    \end{pmatrix}
\end{equation}
where $\Gamma$ is the Gamma function, $\gamma = 0.5772157$ is the Euler's constant and
\begin{equation*}
    p= (1+\xi)^2\Gamma(1+2\xi),\ \ q = (1+\xi)\Gamma'(1+\xi) + (1+\frac{1}{\xi})\Gamma(2+\xi),\ \ r = \Gamma(2+\xi)
\end{equation*}
(See \cite{beirlant2005} and \cite{prescott1980}).

\textbf{PWM}: The asymptotic normality   for the PWM estimator of $\xi$  was proved in \cite{ferreira2015}. Namely the asymptotic variance of $\sqrt{N}(\hat\xi_N^{(PWM)}-\xi)$ is the variance of the gaussian random variable $u Z$ where 
\begin{equation*}
    u = \frac{1}{\Gamma(1-\xi)}\left(\frac{\log 3}{1-3^{-\xi}}-\frac{\log 2}{1-2^{-\xi}}\right)^{-1} \;\;\text{and} \;
     Z = \left(\frac{\xi}{3^\xi-1}(X_2-X_0)-\frac{\xi}{2^\xi-1}(X_1-X_0)\right)
\end{equation*}
Here $X_0, X_1, X_2$ is a gaussian random vector with zero mean and covariance matrix given by 
\begin{align*}   
     Cov(X_r,X_l) =& (l+1)(r+1) \\
     &\int^1_0 \int^1_0 s^{r-1}(-\log s)^{-1-\xi}  u^{l-1}(-\log u)^{-1-\xi} (\min\{s,u\}-su) dsdu
\end{align*}
for $l,r=0,1,2$. See detailed computation of this asymptotic variance in \cite{ferreira2015}.

\textbf{DEH}: {\cite{dekkers1989} introduced a moment estimator of $\xi\in\mathbb R$. To differentiate it from the PWM method, we refer it to DEH estimator here. This estimator is based on the $k(N)$ largest observations from a sample size of $N$.}

\cite{dekkers1989} proved that $\sqrt{k(N)}(\hat\xi^{(DEH)}_{k(N),N}-\xi)$ is asymptotically normal with zero mean and variance
\begin{equation*}
\begin{cases}
   1+\xi^2, &\xi\geq 0 \\
   (1-\xi)^2(1-2\xi)\left[4-8\frac{1-2\xi}{1-3\xi}+\frac{(5-11\xi)(1-2\xi)}{(1-3\xi)(1-4\xi)}\right], &\xi<0
\end{cases}
\end{equation*}




\end{document}